\documentclass{article}
\usepackage{fullpage}

\usepackage{cite}

\usepackage{subfiles}

\usepackage{algorithmic}

\usepackage{amsmath, amssymb,bm, amsfonts}
\usepackage{hyperref}
\usepackage{enumerate,mathtools}
\usepackage{color}
\usepackage{graphicx}

\usepackage{epsfig}
\usepackage{caption}
\usepackage{subcaption}
\usepackage{amsthm}
\usepackage{graphicx}

\newcommand{\bG}{\bm{G}}

\newcommand{\bW}{\bm{W}}
\newcommand{\bX}{\bm{X}}
\newcommand{\bY}{\bm{Y}}
\newcommand{\bZ}{\bm{Z}}
\newcommand{\cov}{\text{Cov}}

\newcommand{\cD}{\mathcal{D}}

\newcommand{\cF}{\mathcal{F}}

\newcommand{\cI}{\mathcal{I}}

\newcommand{\cU}{\mathcal{U}}

\newcommand{\cY}{\mathcal{Y}}

\newcommand{\uG}{\underline{\bG}}

\newcommand{\uZ}{\underline{\bZ}}
\newcommand{\uY}{\underline{\bY}}

\newcommand{\normal}{\mathcal{N}}

\newcommand{\Id}{\mathrm{I}}
\usepackage{array}
\usepackage{tikz}
\usepackage{rotating}
\usepackage{amsmath}

\newcommand{\psd}{\mathbb{S}_+}

\DeclareMathOperator{\gtr}{tr}

\newcommand{\ex}[1]{\ensuremath{\mathbb{E}\left[ #1\right]}}

\newcommand{\diag}{\ensuremath{\text{diag}}}
\DeclareMathOperator{\MMSE}{\mathsf{MMSE}}

\let\originalleft\left
\let\originalright\right
\renewcommand{\left}{\mathopen{}\mathclose\bgroup\originalleft}
\renewcommand{\right}{\aftergroup\egroup\originalright}

\newcommand{\dd}{\mathrm{d}}

\newcommand{\eq}{\begin{equation*}}
\newcommand{\en}{\end{equation*}}

\newcommand{\ea}{\begin{eqnarray*}}
\newcommand{\an}{\end{eqnarray*}}

\newcommand{\ben}{\begin{equation}}

\newcommand{\een}{\end{equation}}

\newtheorem{theorem}{Theorem}
\newtheorem{lemma}[theorem]{Lemma}

\theoremstyle{definition}

\newtheorem{assumption}{Assumption}

\begin{document}

\title{Mutual Information in Community Detection with Covariate Information and Correlated Networks
}

\author{Vaishakhi Mayya \and Galen Reeves}

\maketitle

\begin{abstract}

We study the problem of community detection when there is covariate information about the node labels and one observes multiple correlated networks. We provide an asymptotic upper bound on the per-node mutual information  as well as a heuristic analysis of a multivariate performance measure called the MMSE matrix. These results show that the combined effects of seemingly very different types of information can be characterized explicitly in terms of formulas involving low-dimensional estimation problems in additive Gaussian noise. Our analysis is supported by numerical simulations. 
\end{abstract}

\section{Introduction }\label{intro}
Networks model relational data between various nodes,  e.g., friendship networks in schools or social media. The community detection problem aims to classify the nodes of a network based on those relationships into various communities. The stochastic block model (SBM) is a generative model for a network where each node belongs to exactly one of $k$ communities and the probability of an edge between two nodes is exclusively a function of their community memberships \cite{holland1983stochastic}. In this setting, the goal of community detection is to recover the community labels from the observed network.

 A recent line of work has studied the information-theoretic limits of recovery. Most of this work has focused on either the two-community SBM \cite{decelle2011inference,deshpande:2017,caltagirone:2018, lelarge:2018,barbier:2016a,lesieur:2017,krzakala:2016,deshpande:2018} or the so-called $k$-community symmetric SBM \cite{lesieur:2017, abbe:2018a, banks:2016a,abbe:2018}. 
 In all of these cases,  performance is summarized in terms of a single numerical value, which is often referred to as the effective signal-to-noise ratio of the problem. General SBMs have been considered by Abbe and Sandon~\cite{abbe:2018a} who characterize conditions for weak recovery,  Lesieur et al.~\cite{lesieur:2017} who analyze the performance of an approximate message passing algorithm, and  Reeves et.\ al \cite{rmv2019} who study the asymptotic per-node mutual information and  MMSE in degree-balanced SBMs. 

The contribution of this paper is to extend the analysis in \cite{rmv2019} to the setting where one observes:
\begin{enumerate}
\item covariate information about the node labels; and
    \item multiple networks that are conditionally independent given the same underlying node labels. 
\end{enumerate}
Section~\ref{sec:background} gives the problem formulation and describes connections with previous work.  Section~\ref{sec:theory} provides the main theoretical results, which are upper bounds on mutual information. Numerical simulations are provided in Section~\ref{sec:simulation}. 

\textit{Notation}: We use  $\mathbb{S}^d$, $\mathbb{S}_+^d$ to denote the space of $d \times d$ symmetric matrices and symmetric positive semi-definite matrices, respectively. Given a  positive semi-definite matrix $S$, we use $S^{1/2}$ to denote the unique positive semi-definite square root. Given matrices $A, B \in \mathbb{S}^d$, the relation $A \succeq B$ means that $A - B \in \mathbb{S}_+^d$.

\section{Problem Formulation and Related Work}\label{sec:background}

\subsection{Node labels and covariate information}\label{sec:model}

The labels and covariate information associated with a collection of $n$ nodes are modeled in terms of an i.i.d.\ sequence of  tuples $\{(X_i, Y_i, \tilde{Y}_i)\}_{i=1}^n$ where $X_i$ is the unknown node label and $(Y_i, \tilde{Y}_i)$ is observed covariate information associated with the $i$-th node. 

We focus on the problem of community detection where each label takes exactly one of $k$ values with probability vector $p = (p_1, \dots, p_k)$. Without loss of generality these labels can be embedded into finite dimensional Euclidean space. To facilitate the exposition of our results, we use the whitened representation described in \cite{rmv2019}, where the labels are supported on a set of $k$ points in $\{\mu_1, \dots, \mu_k\}$ in $\mathbb{R}^{k-1}$ with the property that 
\begin{align}
\sum_{a=1}^k p_a \mu_a = 0, \qquad \sum_{a=1}^k p_a \mu_a \mu_{a}^T  = I. \label{eq:mu_white}
\end{align}
A unique specification of this whitened representation is described in \cite[Remark~1]{rmv2019}.

There are two types of the covariate information. The terms $Y_i$ are supported on a set $\cY$ and are used to model general information about the nodes. The terms $\tilde{Y}_i$ correspond to the output of linear Gaussian channel described by
\begin{align}
Y_i = S^{1/2} X_i + N_i,
\end{align}
where $S \in \psd^{k-1}$ is known and $N_i \sim \normal(0,I_{k-1})$ is independent Gaussian noise. These terms play a fundamental role in our proof technique. 

Furthermore, we define the information function $\cI(S) : \psd^{k-1} \to \mathbb R$ and MMSE function $M(S) : \psd^{k-1} \to \psd^{k-1}$ according to
\begin{align}
    \cI(S)  &\triangleq I(X_1; Y_1,  \tilde Y_1)  \label{eqn:I_gen}\\
    M(S) & \triangleq \ex{ \cov(X_1\mid Y_1, \tilde Y_1)},
\end{align}
where $S$ appears in the definition of $\tilde{Y}_i$. The matrix version of the I-MMSE relation \cite{reeves2018mutual} states that
\[
\nabla_S \, \cI(S) = \frac{1}{2} M(S).
\]

Finally, the collection of node labels is represented by an $n \times (k-1)$ matrix $\bX = (X_1, \dots X_n)^T$. Similarly, the covariate information is denoted by matrices $\bY = (Y_1, \dots, Y_n)^T$ and $\tilde{\bY} = (\tilde{Y}_1, \dots, \tilde{Y}_n)^T$ with $\uY = (\bY, \tilde{\bY})$.

\subsection{Correlated networks}

 We consider the setting where one observes multiple networks $\bG_1, \dots, \bG_L$ that are conditionally independent given the labels $\bX$. Each network  is represented by an $n \times n$ binary adjacency matrix $\bG_\ell = (G^\ell_{ i j})$ where $G^\ell_{ij} = G^\ell_{ji} = 1$ if there is an edge between nodes $i$ and $j$ and zero otherwise. Following \cite{rmv2019}, each network is drawn according to a degree-balanced SBM of the form 
\begin{align}\label{eq:GBern}
   G^\ell_{ij} \sim \text{Ber}\left(\frac{d_\ell}{n} + \frac{\sqrt{d_\ell ( 1- d_\ell/n) }}{n} X_i^T R_\ell X_j\right ), \quad i< j,
\end{align}
where $d_\ell$ is a positive real number that parameterizes the expected degree of each node in the network and $R_\ell$ is a symmetric $(k-1) \times (k-1)$ matrix that describes the relationship between the community labels and the probability of an edge. We assume  that the parameters $(d_\ell, R_\ell)$ are known and we use $\uG = (\bG_1, \dots, \bG_L)$ to denote the collection of networks.

\subsection{Multivariate performance metric}

The ability to recover the labels $\bX$ from the observations $(\uY, \uG)$ is assessed in terms of the MMSE matrix:

\begin{align} \label{eqn:MMSE_def}
\MMSE(\bX \mid \uY, \uG) \triangleq \frac{1}{n} \sum_{i=1}^n \ex{ \cov(\bX \mid \uY, \uG)},
\end{align}
where the expectation is taken with respect to $(\uY, \uG)$.  By the matrix I-MMSE relation \cite{reeves2018mutual}, this matrix can also be expressed as the gradient of the mutual information with respect to the matrix SNR:
 \[
 \MMSE(\bX \mid \uY, \uG) = 2 \nabla_S   I(\bX; \uY, \uG).
 \]
Moreover, by the data processing inequality for covariance and the assumption that the rows of $\bX$  drawn from the whitened representation, 
$ 0 \preceq \MMSE(\bX \mid \uG,\uY) \preceq  \Id_{k-1}.$

Notice that in the absence of network observations $\uG$, the problem of estimating $\bX$ from the covariate information $\uY$ decouples into $n$ independent problems and we have:
\begin{align}
\frac{1}{n} I(\bX; \uY) &= \cI(S)\\
\MMSE(\bX\mid  \uY) &= M(S).
\end{align}
These terms involve $(k-1)$-dimensional integrals that can be approximated numerically for small values of $k$.  The problem of estimating the node labels in the presence of network observations is more difficult to analyze because the networks induce dependence in the conditional distribution of the labels.

\subsection{Relation to prior work}

A great deal of recent work has used ideas from from information theory and statistical physics to characterize the information-theoretic limits of community detection (from a single network) as well as the performance of computationally efficient methods, \cite{decelle2011inference,abbe2015detection,deshpande:2017,caltagirone:2018, lelarge:2018,barbier:2016a,lesieur:2017,krzakala:2016,deshpande:2018,rmv2019,moore2017computer, abbe2017community}.  Much of this work has focussed on the weak recovery problem, which requires that the community labels are estimated with a mean-squared error that is strictly better than that of random guessing \cite[Chapter 4]{abbe2017community}.  On the algorithmic side, it has been shown that weak recovery is possible using polynomial time algorithms provided that the matrix $R$ has at least one eigenvalue with magnitude greater than one~\cite{bordenave2015non, krzakala2013spectral, abbe2015detection}. This condition is sometimes referred to as the Kesten-Stigum (KS) threshold. The information-theoretic limits describe the optimal performance that can be attained without any constraints on computational complexity. For network models with $k \ge 4$ communities \cite{decelle2011inference,lelarge:2018,coja2018information} or asymmetries \cite{lelarge:2018,rmv2019}, there exists a computational-to-statistical gap, where weak recovery below the Kesten-Stigum threshold is information-theoretically possible, even though all known polynomial-time algorithms fail in this regime.

The main results of this paper apply to the so-called dense network setting where the expected degree  $d$ of each node in the network increases with the problem dimension $n$. In this setting, previous work has provided bounds on the  asymptotic minimum mean-squared error of estimating the community labels \cite{deshpande:2017,caltagirone:2018, lelarge:2018,barbier:2016a,lesieur:2017,krzakala:2016,deshpande:2018,rmv2019}. The analysis in this paper builds upon the recent work in \cite{rmv2019}, which shows that the mutual information and MMSE in a degree balanced SBM can be characterized in terms of a matrix of effective signal-to-noise ratios.

The impact of covariate information (also known as side information) has been studied previously \cite{newman2016structure,cai2016inference,stegehuis2019efficient,kanade2016global,mossel2016local, ver2014phase,zhang2014phase,binkiewicz2017covariate,saad2018community,deshpande2018contextual}.  In some cases \cite{stegehuis2019efficient,kanade2016global,zhang2014phase} it has been shown that relatively small amount of node-wise covariate information can have a large impact on performance and also significantly reduce the computational-to-statistical gap. Much of the theoretical analysis \cite{cai2016inference,stegehuis2019efficient,kanade2016global,mossel2016local, ver2014phase,zhang2014phase}   has focussed on the 2-community or symmetric SBM. A contribution of this paper is to consider the larger class of degree balanced SBMs.

There has also been some recent work on community detection with multiple correlated networks~\cite{levin2017central,arroyo2019inference}, which focuses on scaling regimes where the  eigenvalues of $R_\ell$ scale with size of network $n$, and thus the ability to detect communities improves as  $n$ goes to infinity. In contrast, this paper  focuses on the setting where $R_\ell$ is a constant  and thus the mean-squared error is non-vanishing. To the best of our knowledge, the information theoretic limits  for  community detection with multiple correlated networks have not been addressed.

\section{Formulas for Mutual information and MMSE}\label{sec:theory}

\subsection{Upper bound on the mutual information}

Our analysis focuses on a sequence of problem settings where the number of nodes $n$ scales to infinity. We assume that node labels and covariate information are drawn i.i.d.\ according to the distribution on  $(X_1, Y_1, \tilde{Y_1})$ and the matrices $\{R_\ell\}$ are fixed. We make two additional assumptions. 

\begin{assumption}[Diverging Average Degree]\label{assump:lin_deg} The average degree of each network $d_\ell$ increases with $n$ such that both $d_\ell$ and $(n-d_\ell)$ tend to infinity. 
\end{assumption}

\begin{assumption}[Definite Matrix]\label{assump:defR} Each matrix  $R_\ell$ is either positive definite or negative definite.  
\end{assumption}

Our results are stated in terms of a potential function. 
Let $\cU = \{ U \in \psd^{k-1} \, : \, U \preceq I\}$ and let  $\cF: \cU  \to [0, \infty)$ be defined as 
\begin{align}\label{eqn:fpmain}
\cF(U) & \triangleq \cI\Big(S + \sum_{\ell=1}^L R_\ell (I - U) R_\ell \Big) + \frac{1}{4} \sum_{\ell=1}^L \gtr( (R_\ell U)^2). 
\end{align}

The following result provides an asymptotic upper bound on the per-node mutual information between $\bX$ and the observations $(\uY, \uG)$. The proof is given in Section~\ref{thm:I_UB:proof}.   

\begin{theorem}\label{thm:I_UB}
Under Assumptions~\ref{assump:lin_deg} and \ref{assump:defR},  
\begin{align}\label{eqn:I_main}
\limsup_{n \to \infty} \frac{1}{n} I(\bX ; \uY, \uG) \le  \min_{U \in \cU} \cF(U). 
\end{align}
\end{theorem}

Theorem~\ref{thm:I_UB} provides an extension of \cite{rmv2019}, which focused on the setting of a single network ($L=1$) without the covariate information provided by $\bY$. In this setting, \cite[Theorem~1]{rmv2019} shows that the upper bound is asymptotically tight when $S = 0$, that is 
\begin{align}
\lim_{n \to \infty} \frac{1}{n}I(\bX; \uG) = \min_{U \in \cU} \left\{  \cI(  R (I - U) R) + \frac{1}{4} \gtr( (R U)^2) \right\}.
\end{align}

\subsection{Partially revealed labels} 
\label{sec:partial_labels}
As a specific example of covariate information, consider the setting where a fraction of the true node labels are revealed. This is also referred to as the semi-supervised setting \cite{zhang2014phase}. Using the setup introduced in Section~\ref{sec:model}, partially revealed labels can be modeled using an erasure channel, where $Y_i$ is equal to $X_i$ with probability $\alpha$ and is equal to an erasure symbol with probability $1-\alpha$. In this setting, the mutual information function is given by
\begin{align}
\cI(S)  = \alpha H(X_1) + (1- \alpha) I(X_1; \tilde{Y}_1)
\end{align}
where $H(X_1) = \sum_{a=1}^k - p_a \log p_a$ is the entropy of the community labels.

\subsection{Heuristic analysis of MMSE matrix}

The MMSE matrix is related to the mutual information via the matrix I-MMSE relation~\cite{reeves2018mutual}, which implies
\begin{align*}
I(\bX; \uY, \uG)  - I(\bX; \bY, \uG) 
&
= \frac{n}{2} \int_0^1 \gtr\left( \MMSE(\bX \mid \uG, \uY) \Big\vert_{S = S_\gamma}  \frac{\dd}{ \dd \gamma} S_\gamma \right) \, \dd \gamma,
\end{align*}
for any differentiable path $S_\gamma$ with $S_0 = 0 $ and $S_1 = S$. Following the approach outlined in \cite[Appendix A.3]{rmv2019}, it can be shown that upper and lower bounds on the asymptotic per-node mutual information lead to asymptotic bounds on the MMSE matrix. In particular, for the special case of a single network without covariate information,  \cite[Theorem~3]{rmv2019} shows that, for any positive definite $S$, 
\[
\MMSE(\bX \mid  \tilde{\bY},  \bG) \preceq U^* + o_n(1),
\]
where $U^*$ is any minimizer of $\cF(U)$ and $o_n(1)$ denotes a sequence of symmetric matrices that converges to zero in the large-$n$ limit. 

Our next result follows a similar approach for the setting of multiple networks and covariate information. This result requires the additional assumption that the upper bound on the mutual information in Theorem~\ref{thm:I_UB} is asymptotically tight for $S = 0$. Because this assumption is unproven, the resulting upper bound is considered to be heuristic.

\begin{theorem}\label{thm:AUB}
Consider Assumptions 1 and 2. If the upper bound in Theorem~\ref{thm:I_UB} is asymptotically tight at $S = 0$, that is
\begin{align*}
\lim_{n \to \infty} \frac{1}{n}I(\bX; \uG,\bY) = 
 &\min_{U \in \cU} \left\{  \cI\Big( \sum_{\ell=1}^L R_\ell (I - U) R_\ell \Big) + \frac{1}{4} \sum_{\ell=1}^L \gtr( (R_\ell U)^2) \right\}
\end{align*}
then, for any positive definite $S$, the MMSE matrix satisfies
\begin{equation} \label{eqn:mmseaub}
\MMSE(\bX \mid \uG,\uY) \preceq U^* + o_{n}(1), 
\end{equation}
where $U^*$ is any minimizer of $\cF(U)$ and $o_n(1)$ denotes a sequence of symmetric matrices that converges to zero in the large-$n$ limit.
\end{theorem}

\section{Simulation Results}\label{sec:simulation}

\subsection{Covariate information} \label{sec:covinf}

We first consider the effects of partially revealed labels in the setting of a single network observation.  Results are obtained on a problem with $n=10^5$ nodes and $k=3$ communities with probability vector  $p = (0.1, 0.3, 0.6)$. Conditional on the node labels, the network is drawn according to a degree-balanced SBM with average degree $d =30$ and $R = \diag(\lambda_1, \lambda_2)$. The covariate information in $\bY$ consists of the output of an erasure channel, as described in Section~\ref{sec:partial_labels}.

We compare our theoretical results  with the empirical performance of belief propagation (BP). For each problem setting, the MSE is estimated according to $\frac{1}{n} \sum_{i=1}^n \| X_i - \hat{X}_i\|^2,$ where $\hat{X}_i$ is the BP estimate of the $i$-th label. We note that this evaluation of the MSE differs slightly from much of the prior work, which focuses on uniform community assignments and includes an additional step that minimizes over all permutations of community labels. This additional step is not needed in our setting due to the non-uniformity in community sizes.

\begin{figure*}
\centering
\begin{subfigure}[b]{0.45\textwidth}
\includegraphics[width =\linewidth]{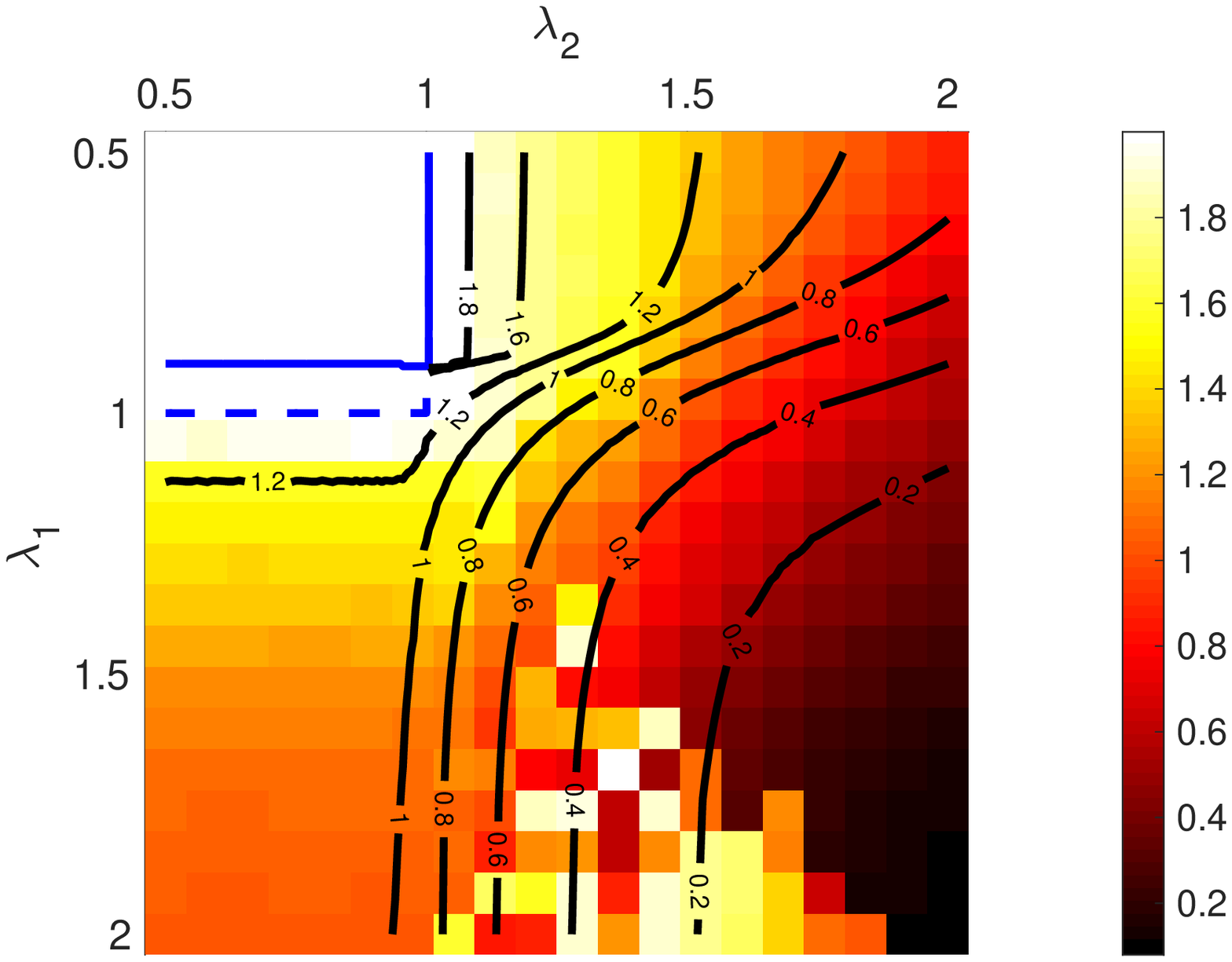}
\caption{No covariate information}
\end{subfigure}
\begin{subfigure}[b]{0.45\textwidth}
\includegraphics[width =\linewidth]{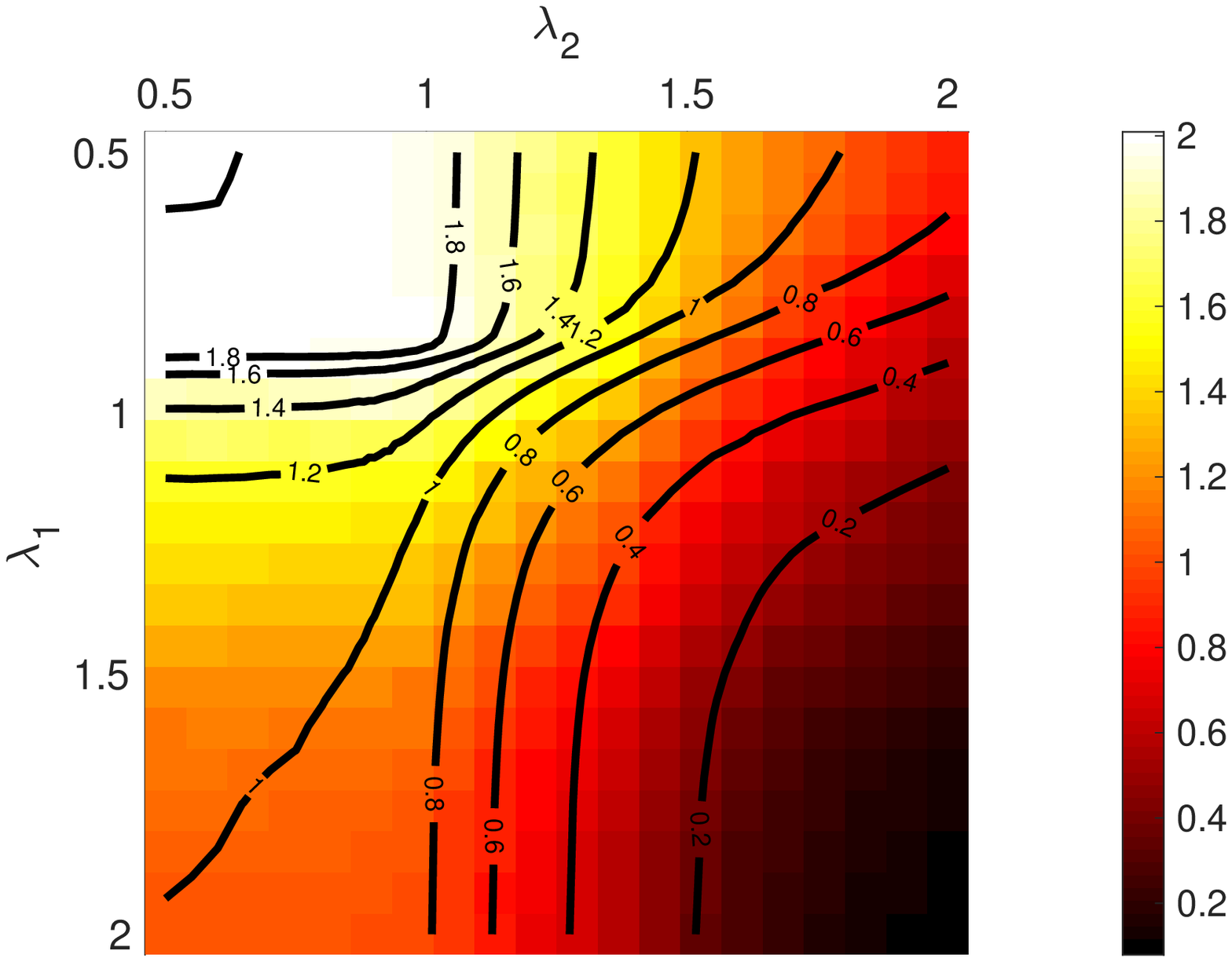}
\caption{$1\%$ of labels revealed}
\end{subfigure}
\caption{Comparison of the heuristic upper bound on $\gtr(\MMSE(\bX \mid \uY,\bG)$ given in Theorem~\ref{thm:AUB} (black contour lines) with the empirical MSE of BP (heat map). In the left panel, the solid blue line is the upper bound on the weak recovery threshold given in \cite[Theorem~5]{rmv2019} and the dashed blue line is the weak recovery threshold for acyclic BP \cite{abbe:2018a}.  
\label{fig:compare}}
\end{figure*}

Figure~\ref{fig:compare} provides a comparison of 
the heuristic upper bound on $\gtr(\MMSE(\bX \mid \uY,\bG)$ given in Theorem~\ref{thm:AUB}  and the empirical MSE of BP, where each pixel is the median of 8 independent trials. The axes correspond to the eigenvalues of $R$. Figure~\ref{fig:compare}(a) corresponds to the setting without covariate information and  Figure~\ref{fig:compare}(b) corresponds to the setting where $1\%$ of the labels are revealed.

Similar to previous work focusing on partially revealed labels~\cite{zhang2014phase,cai2016inference, saad2018community,stegehuis2019efficient}, Figure~\ref{fig:compare} shows  that a relatively small amount of extra information can provide significant performance gains. One of main takeaways from Figure~\ref{fig:compare} is that there is a close qualitative correspondence between the heuristic upper bound given in this paper and the empirical performance. 

Finally,  we note that there is a region in  Figure~\ref{fig:compare}(a) where BP becomes unstable. We suspect that this may be a consequence of asymmetries in the network model.

\subsection{Correlated networks} \label{sec:multnet}

Next, we consider the effects of multiple network observations.  Results are obtained for a problem with $n=10^4$ nodes and $k=3$ communities with non-uniform probability vector  $p = (0.1, 0.3, 0.6)$. Conditional on the labels, two networks are drawn according to the degree-balanced SBM with average degree $d = 30$ and $R_\ell = r I_2$.

In this setting, we found that the BP has convergence issues and so we compare our theoretical results with the empirical performance of a spectral method \cite{mmvr2019} applied to a linear combination of the adjacency matrices. Specifically, we obtain estimates of the community labels using the following procedure. First, we construct the average of the networks $\bG_1$ and $\bG_2$ according to 
\begin{align}
\tilde{\bG} = \frac{1}{\sqrt{2}}\bG_1 + \frac{1}{\sqrt{2}} \bG_2.
\end{align}
Note that the conditional expectation of $\tilde{\bG}$ given $\bX$ is comparable to that of a single network with $\tilde{R} = \sqrt{2 }r \Id$. Next, we retain the eigenvectors associated with the second and third leading eigenvalues in the spectral decomposition of $\tilde{\bG}$. The relationship between these eigenvectors and the node labels is characterized using a Gaussian mixture model (GMM) approach described in \cite{mmvr2019}, evaluated with $\tilde{R}$.

\begin{figure}
\centering
\includegraphics[width =0.5\linewidth]{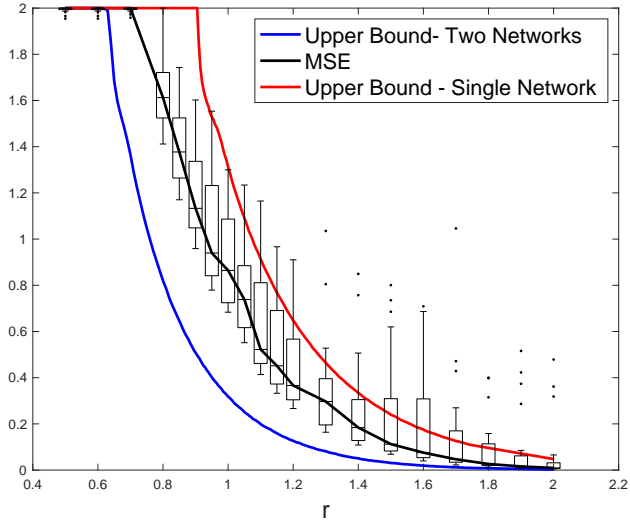}
\caption{
MSE as a function of the SBM parameter $r$.
}
\label{fig:cov_net}
\end{figure}

Figure~\ref{fig:cov_net} shows the MSE as a function of the SBM parameter $r$.  The solid blue line corresponds to the trace of the heuristic upper bound to the MMSE for two correlated networks computed from Theorem~\ref{thm:AUB}, and the red line corresponds to the upper bound for a single network.  The black line corresponds to the empirical observations using the method described in this section.  With multiple correlated networks, we see that the MSE shows an improvement in the presence of additional information, and our proposed asymptotic upper bound follows the observed performance.

 \section{Conclusion}

 In this work, we  study the information-theoretic limits of community detection with covariate information and correlated networks. Our main result (Theorem~\ref{thm:I_UB}) provides and upper bound on the mutual information between the community labels $\bX$ and the observed data, which consists of the collection of graphs $\uG$ and the side information $\uY$ when the model parameters are known. Next, we leverage the multivariate I-MMSE relationship to provide an upper bound on  MMSE in estimating the community labels (Theorem~\ref{thm:AUB}). This result requires the assumption that our upper bound in Theorem~\ref{thm:I_UB} is tight. Our theoretical analysis is supported by the numerical simulations demonstrating the impact of partially revealed labels (Section~\ref{sec:covinf}) and multiple correlated networks (Section~\ref{sec:multnet}).

An important direction for future work is to determine whether the upper bound given in Theorem~\ref{thm:I_UB} is tight.  We note that such a result has been established previously special case of a single network without side information \cite[Theorem 2]{rmv2019}.

\section{Proof of Theorem~\ref{thm:I_UB}}\label{thm:I_UB:proof}

The proof of Theorem~\ref{thm:I_UB} follows the approach in \cite{rmv2019} with appropriate modifications to handle the covariate information and multiple networks. The first step of the proof is to establish an asymptotic equivalence between the mutual information in the community detection problem and the mutual information in the symmetric matrix estimation problem defined by
\begin{align}
    \bW_\ell &= \frac {1}{\sqrt n}\bX R_\ell \bX^T  \label{eqn:W}\\
    \bZ_\ell &= \sqrt t \bW_\ell + \bm{\xi}_\ell \label{eqn:LRME},
\end{align}
where $\bm{\xi}$ is a symmetric matrix with $\xi_{ij}  \sim \normal(0,1)$ for $i<j$ and $\xi_{ii} \sim \normal(0,2)$. We use $\uZ = (\bZ_1, \dots, \bZ_L)$ to denote the collection of matrix observations.

 \begin{lemma}[Channel Universality]\label{lem:universality}
 Under Assumption~\ref{assump:lin_deg},
\begin{align}
    \lim_{n \to \infty} \frac{1}{n} \left \lvert I(\bX;\uY, \uG) - I(\bX;\uY, \uZ)\right \rvert = 0.
\end{align}
\end{lemma}
\begin{proof}
To simplify the expression,  we will prove the result without $\bY$. The result can then be extended to the setting with $\bY$ following the approach used in \cite[Corollary~7]{rmv2019}.

To proceed, let us define $a_1 = I(\bX; \uZ)$, $a_{L+1} = I(\bX ; \uG)$, and 
\begin{align*}
a_\ell = I(\bX; \bG_1, \dots, \bG_{\ell-1}, \bZ_{\ell}, \dots, \bZ_{L} ), 
\end{align*}
for $\ell = 2, \dots, L$. By the triangle inequality, we can then write
\begin{align*}
    \left| I(\bX; \uG) - I(\bX; \uZ) \right|  & = 
    \left| \sum_{\ell=1}^{L} a_{\ell + 1}  - a_\ell  \right|  \le  \sum_{\ell=1}^{L}\left| a_{\ell + 1}  - a_\ell  \right|.  
\end{align*}

Next, by the chain rule for mutual information one finds that
\begin{align*}
a_{\ell +1} - a_\ell & = I(\bX; \bG_\ell \mid \cD_\ell) -  I(\bX; \bZ_\ell \mid \cD_\ell)\\
&= I(\bW_\ell; \bG_\ell \mid \cD_\ell) -  I(\bW_\ell; \bZ_\ell \mid \cD_\ell),
\end{align*}
where $\cD_\ell = (\bG_1, \dots, \bG_{\ell-1}, \bZ_{\ell+1}, \dots \bZ_{L})$. 
Under the assumed distribution on $\bW_\ell$, we can apply \cite[Theorem~6]{rmv2019} to show that $\frac{1}{n} |a_{\ell + 1} - a_\ell|$ converges to zero in the large-$n$ limit. \end{proof}

The next step in our proof is to obtain an upper bound on $I(\bX ; \uY, \uZ)$. We define the function 
\begin{align}
    \cI(S,t) \overset{\Delta} =  \frac{1}{n} I(\bX;\tilde\uY,\uZ). \label{eqn:inf_func}
\end{align}
where we note that $\cI(S,0) = \cI(S)$ is the information function defined in \eqref{eqn:I_gen}.  The function $\cI(S,t)$ is concave and differentiable in $(S,t)$ with 

\begin{align}
    \nabla_S \cI(S,t) & = \frac{1}{2} \MMSE(\bX \mid \uY, \uZ).
\end{align}
The next result provides an upper bound on the partial derivative with respect to $t$.

\begin{lemma}\label{lem:grad_inq}
Under Assumption~\ref{assump:defR},
\begin{align}
   \partial_t \cI(S,t) & \le \frac{1}{4} \sum_{\ell=1}^L g_{\ell}(2 \nabla_S \cI(S,t))
\end{align}
where
\begin{align*}
  g_{\ell}(U)= \frac{1}{n^2}\gtr \left ({\ex{(R_\ell\bX^T\bX)^2}}\right) -   \gtr\left ( (R_\ell(I - U))^2\right ).
\end{align*}
\end{lemma}
\begin{proof}
Suppose that each observation $\bZ_\ell$ has a separate parameters $t_\ell$. By the chain rule for differentiation, we can then write
\begin{align}
  \partial_t \cI(S,t)  & = \sum_{\ell=1}^L \partial_{t_\ell} \frac{1}{n} I(\bX; \uY, \uZ) \Big \vert_{t_1 = \dots t_L = t}.
\end{align}
Furthermore, by the chain rule for mutual information and the fact that $\bZ_\ell$ is conditionally independent of everything else given $\bW_\ell$, we have

\begin{align*}
  \partial_{t_\ell}I(\bX; \uY, \uZ)= \partial_{t_\ell}  I(\bW_\ell ; \bZ_\ell \mid \uY,  \bZ_{\sim \ell} ), 
\end{align*}
where the subscript $\sim \ell$ means that the $\ell$-th term is omitted.

Following the steps outlined in  outlined in  \cite[Appdendix~D]{rmv2019} and the proof of \cite[Lemma~11]{rmv2019}, one finds that the $\partial_{t_\ell}\cI(S,t)\le \frac{1}{4} g_\ell(\nabla_S \cI(S,t))$. Plugging this inequality back into the expression above completes the proof. 
\end{proof}

Having established Lemma~\ref{lem:grad_inq}, the rest of the proof follows similarly to the proof of Theorem~8 in~\cite{rmv2019}. Specifically, we obtain
\begin{align}
\cI(S,1) & \le \min_{U \in \cU}\left\{\cI^*(U)  + \frac{1}{2} \gtr( S U) + \frac{1}{4} \sum_{\ell=1}^L g_{\ell}(U)   \right\}, \label{eq:conj_UB}
\end{align}
where
\begin{align}
\cI^*(U) & = \sup_{S \succeq 0} \left\{ \cI(S) - \frac{1}{2} \gtr( S U) \right\}
\end{align}
is the convex conjugate of $\cI(S)$. 

For the final step in the proof, observe that
\begin{align*}
    g_\ell(U) =  \delta_\ell  + 2 \gtr( R_\ell^2 U ) -    \gtr\left ( (R_\ell U)^2\right ),
\end{align*}
where  $\delta_\ell  =\frac{1}{n^2}\gtr \left ({\ex{(R_\ell(\bX^T\bX - I))^2}}\right)$. 
For all $\tilde{U} \in \cU$, the inequality 
\begin{align*}
- \gtr( (R_\ell U)^2)
&\le    -2 \gtr(R_\ell U R_\ell \tilde{U})  + \gtr( (R_\ell \tilde U)^2),
\end{align*}
leads to 
\begin{align*}
    g_\ell(U) \le  \delta_\ell    + 2 \gtr( R (I - \tilde{U}) R U)     + \gtr( (R_\ell \tilde{U})^2).
\end{align*}
Combining this inequality with \eqref{eq:conj_UB}, we see that, for all $U, \tilde{U}$ in $\cU$,
\begin{align}
\cI(S,1) & \le  \cI^*(U)  + \frac{1}{2} \sum_{\ell=1}^L\gtr( (S  + R_\ell(I-\tilde{U})R_\ell) U) 
+ \frac{1}{4} \sum_{\ell=1}^L \gtr((R_\ell \tilde{U})^2) +  \frac{1}{4} \sum_{\ell=1}^L  \delta_\ell . \label{eq:I_U_Ut}
\end{align}
The minimum of the first two terms with respect to $U$ then leads to 
\begin{align*}
\min_{U \in \cU} \left\{  \cI^*(U)  + \frac{1}{2} \sum_{\ell=1}^L \gtr( (S  + R_\ell(I-\tilde{U})R_\ell) U) \right\} 
& = \cI\Big( S  + \sum_{\ell = 1}^L R_\ell(I-\tilde{U})R_\ell \Big).
\end{align*}
where we have used the fact that $\cI(S)$ is concave, and thus equal to its biconjugate. Plugging this expression back into \eqref{eq:I_U_Ut} and then taking the the minimum with respect to $\tilde{U}$ yields,
\begin{align} \label{eq:cI_UB_c}
    \cI(S,1) \le \min_{\tilde{U} \in \cU}   \cF(\tilde U) +  \frac{1}{4} \sum_{\ell=1}^L  \delta_\ell.
\end{align}
Under the assumed distribution on $\bX$ each term $\delta_\ell$ vanishes in the large-$n$ limit. Combining \eqref{eq:cI_UB_c} with Lemma~\ref{lem:universality} completes the proof of  Theorem~\ref{thm:I_UB}.

\bibliography{all} 
\bibliographystyle{ieeetr}
\end{document}